\documentclass[a4paper,USenglish,cleveref]{lipics-v2021}
\title{A Simple Algorithm for Graph Reconstruction}

\author{Claire Mathieu}{CNRS Paris, France \and \url{https://www.irif.fr/~claire/}}{claire.mathieu@irif.fr}{}{}
\author{Hang Zhou}{École Polytechnique, France \and \url{http://www.normalesup.org/~zhou/}}{hzhou@lix.polytechnique.fr}{}{}
\authorrunning{C. Mathieu and H. Zhou}
\Copyright{Claire Mathieu and Hang Zhou}


 \nolinenumbers 
\hideLIPIcs  

\acknowledgements{This work was partially funded by the grant ANR-19-CE48-0016 from the French National Research Agency (ANR).
We want to thank the anonymous reviewers for their valuable comments.}

\usepackage{algorithm}
\usepackage[noend]{algpseudocode}
\usepackage{xspace}

\newtheorem{fact}[theorem]{Fact}
\newcommand{\query}{\textsc{Query}}
\newcommand{\expect}[2]{\mathbb{E}_{#1}\big[#2\big]}
\newcommand{\proba}[2]{\mathbb{P}_{#1}\big[#2\big]}
\newcommand{\algoname}{\textsc{Simple}\xspace}
\newcommand{\valgoname}{\textsc{Multi-Phase}\xspace}
\newcommand{\variantname}{\textsc{Simple-Modified}\xspace}
\DeclareMathOperator{\polylog}{polylog}

\newcommand{\apref}[1]{\cref{#1}}

\begin{document}

\maketitle

\begin{abstract}
How efficiently can we find an unknown graph using distance queries between its vertices? We assume that the unknown graph is connected, unweighted, and has bounded degree. The goal is to find every edge in the graph. This problem admits a reconstruction algorithm based on multi-phase Voronoi-cell decomposition and using $\tilde O(n^{3/2})$ distance queries~\cite{kannan2018graph}.

In our work, we analyze a simple reconstruction algorithm. We show that, on random $\Delta$-regular graphs, our algorithm uses $\tilde O(n)$ distance queries. As by-products, we can reconstruct those graphs using $O(\log^2 n)$ queries to an all-distances oracle or $\tilde O(n)$ queries to a betweenness oracle, and we bound the metric dimension of those graphs by $\log^2 n$.

Our reconstruction algorithm has a very simple structure, and is highly parallelizable. On general graphs of bounded degree, our reconstruction algorithm has subquadratic query complexity.
\end{abstract}

\section{Introduction}
Discovering the topology of the Internet is a crucial step for building accurate network models and designing efficient algorithms for Internet applications.
The topology of Internet networks is typically investigated at the router level, using \texttt{traceroute}.
It is a common and reasonably accurate assumption that \texttt{traceroute} generates paths that are shortest in the network.
Unfortunately, sometimes routers block \texttt{traceroute} requests due to privacy and security concerns.
As a consequence, the inference of the network topology is rather based on the end-to-end delay information on those requests, which is roughly proportional to the shortest-path distances in the network.

In the \emph{graph reconstruction} problem, we are given the vertex set $V$ of a hidden connected, undirected, and unweighted graph and have access to information about the topology of the graph via an oracle, and the goal is to find every edge in $E$.
Henceforth, unless explicitly mentioned, all graphs studied are assumed to be connected.
This assumption is standard and shared by almost all references on the subject, e.g., \cite{beerliova2006network,erlebach2006network,kannan2018graph,rong2021reconstruction,sen2010covert}.
The efficiency of an algorithm is measured by the \emph{query complexity}, i.e., the number of queries to the oracle.
Motivated by \texttt{traceroute}, the literature has explored several types of query oracles.
\begin{itemize}
    \item
    One type consists of \emph{all-shortest-paths} and \emph{all-distances} queries, when querying a vertex yields either shortest paths from that vertex to all other vertices~\cite{beerliova2006network,sen2010covert} or distances from that vertex to all other vertices~\cite{erlebach2006network}. The latter, of course, is less informative.
    \item A more refined type of query oracles, suggested in~\cite{beerliova2006network,erlebach2006network}, consists of  \emph{shortest-path} and \emph{distance} queries, when querying a pair of vertices yields  either a shortest path or the distance between them~\cite{kannan2018graph,reyzin2007learning,rong2021reconstruction}. Again, the latter is less informative.
\end{itemize}

In this work, we focus on the weakest of those four query oracles, that takes as input a pair of vertices $a$ and $b$ and returns the distance $\delta(a,b)$ between them.
Reyzin and Srivastava~\cite{reyzin2007learning} showed that graph reconstruction requires $\Omega(n^2)$ distance queries on general graphs, so we focus on the bounded degree case.
For graphs of bounded degree, Kannan, Mathieu, and Zhou~\cite{kannan2018graph} gave a reconstruction algorithm based on multi-phase Voronoi-cell decomposition and using $\tilde O(n^{3/2})$ distance queries, and raised an open question of whether  $\tilde O(n)$ is achievable.\footnote{The notation $\tilde O(f(n))$ stands for $O(f(n)\cdot \polylog f(n))$.}

We provide a partial answer to that open question by analyzing a simple reconstruction algorithm (\cref{algo:main}).
We show that, on (uniformly) random $\Delta$-regular graphs, where every vertex has the same degree $\Delta$, our reconstruction algorithm uses $\tilde O(n)$ distance queries (\cref{thm:random}).
As by-products, we can reconstruct those graphs using $O(\log^2 n)$ queries to an all-distances oracle (\cref{cor:all-distances}) or using $\tilde O(n)$ queries to a betweenness oracle (\cref{cor:betweenness}), and we bound the metric dimension of those graphs by at most $\log^2 n$ (\cref{cor:metric_dimention}).

Our analysis exploits the \emph{locally tree-like} property of random $\Delta$-regular graphs, meaning that these graphs contain a small number of short cycles.
Our method might be applicable to other locally tree-like graphs, such as Erd\"os-R\'enyi random graphs and \emph{scale-free} graphs.
In particular, many real world networks, such as Internet networks, social networks, and peer-to-peer networks, are believed to have scale-free properties~\cite{barabasi1999emergence,jovanovic2001modeling,newman2002random}.
We defer the reconstruction of those networks for future work.

Our reconstruction algorithm has a very simple structure, and is highly parallelizable (\cref{cor:parallel}).
On general graphs of bounded degree, the same reconstruction algorithm has subquadratic query complexity (\cref{thm:general}).

\subsection{Related Work}
The problem of reconstructing a graph using queries that reveal partial information has been extensively studied in different contexts and has many applications.

\paragraph*{Reconstruction of Random Graphs}
The gist of our paper deals with random graphs.
The graph reconstruction problem has already attracted much interest in the setting of random graphs.
On Erd\"os-R\'enyi random graphs, Erlebach, Hall, and Mihal’{\'a}k~\cite{erlebach2007approximate} studied the approximate network reconstruction using all-shortest-paths queries;
Anandkumar, Hassidim, and Kelner~\cite{anandkumar2013topology} used end-to-end measurements between a subset of vertices to approximate the network structure.
Experimental results to reconstruct random graphs using shortest-path queries were given in~\cite{blondel2007distance,guillaume2005complex}.

On random $\Delta$-regular graphs, Achlioptas~et~al.~\cite{achlioptas2009bias} studied the bias of \texttt{traceroute} sampling in the context of the network reconstruction.
They showed that the structure revealed by \texttt{traceroute} sampling on random $\Delta$-regular graphs admits a power-law degree distribution~\cite{achlioptas2009bias}, a common phenomenon as in  Erd\"os-R\'enyi random graphs~\cite{lakhina2003sampling} and Internet networks~\cite{faloutsos1999}.

\paragraph*{Metric Dimension and Related Problems}
Our work yields an upper bound on the \emph{metric dimension} of random $\Delta$-regular graphs.
The metric dimension problem was first introduced by Slater~\cite{slater1975leaves} and Harary and Melter~\cite{harary1976metric}, see also \cite{bailey2011base,caceres2007metric,chartrand2000resolvability,hernando2010extremal,khuller1996landmarks,oellermann2007strong,ramirez2016simultaneous,sebHo2004metric}.
The metric dimension of a graph is the cardinality of a smallest subset $S$ of vertices such that every vertex in the graph has a unique vector of distances to the vertices in $S$.
On regular graphs, the metric dimension problem was studied in special cases~\cite{chartrand2000resolvability,javaid2008families}.
In Erd\"os-R\'enyi random graphs, the metric dimension problem was studied by Bollob{\'a}s, Mitsche, and Pra{\l}at~\cite{bollobas2013metric}.
Mitsche and Ru\'e~\cite{mitsche2015limiting} also considered the random forest model.

A related problem is the \emph{identifying code} of a graph~\cite{karpovsky1998new}, which is a smallest subset of vertices such that every vertex of the graph is uniquely determined by its neighbourhood within this subset.
The identifying code problem was studied on random $\Delta$-regular graphs~\cite{foucaud2012bounds} and on Erd\"os-R\'enyi random graphs~\cite{frieze2007}.
Other related problems received attentions on random graphs as well, such as the
\emph{sequential metric dimension}~\cite{odor2020sequential} and the \emph{seeded graph matching}~\cite{mossel2020seeded}.

\paragraph*{Betweenness Oracle}
There exists an oracle that is even weaker than the distance oracle: the \emph{betweenness} oracle~\cite{abrahamsen2016graph}, which receives three vertices $u$, $v$, and $w$ and returns whether $w$ lies on a shortest path between $u$ and $v$.
Our work yields a reconstruction algorithm using $\tilde O(n)$ betweenness queries on random $\Delta$-regular graphs.
For graphs of bounded degree, Abrahamsen et al.~\cite{abrahamsen2016graph} generalized the $\tilde O(n^{3/2})$ result in the distance oracle model from \cite{kannan2018graph} to the betweenness oracle model.

\paragraph*{Tree Reconstruction and Parallel Setting}
Our paper focuses on the distance oracle and bounded degree, and considers the parallel setting.
All of those aspects were previously raised in the special case of the \emph{tree reconstruction}.
Indeed, motivated by the reconstruction of a phylogenetic tree in evolutionary biology, the tree reconstruction problem using a distance oracle is well-studied~\cite{hein1989optimal,king2003complexity,waterman1977additive}, in particular assuming bounded degree~\cite{hein1989optimal}.
Afshar et al.~\cite{afshar2020reconstructing} studied the tree reconstruction in the parallel setting, analyzing both the \emph{round complexity} and the \emph{query complexity} in the relative distance query model~\cite{kannan1996determining}.

\subsection{Our Results}
\label{sec:results}
Our reconstruction algorithm, called \algoname, is given in \cref{algo:main}. It takes as input the vertex set $V$ of size $n$ and an integer parameter $s\in[1,n]$.

\begin{algorithm}
\caption{\algoname$(V,s)$}
\label{algo:main}
\begin{algorithmic}[1]
    \State $S\gets$ sample of $s$ vertices selected uniformly and independently at random from $V$
    \For{$u\in S$ and $v\in V$}
         \query$(u,v)$
    \EndFor
    \State $\hat E\gets$ set of vertex pairs $\{a,b\}\subseteq V$ such that, for all $u\in S$, $|\delta(u,a)-\delta(u,b)|\leq 1$
    \For{$\{a,b\}\in \hat E$}
         \query$(a,b)$
    \EndFor
    \State \Return set of vertex pairs $\{a,b\}\in \hat E$ such that $\delta(a,b)=1$ \label{line:return}
\end{algorithmic}
\end{algorithm}
Intuitively, the set $\hat E$ constructed in \algoname consists of all vertex pairs $\{a,b\}\subseteq V$ that \emph{might} be an edge in $E$.
In order to obtain the edge set~$E$, it suffices to query uniquely the vertex pairs in~$\hat E$.
We remark that \algoname correctly reconstructs the graph for any parameter $s \in [1,n]$, and that choosing an appropriate $s$ only affects the query complexity, see \cref{lem:query}.

\subsubsection{Random Regular Graphs}
Our first main result shows that \algoname (\cref{algo:main}) uses $\tilde O(n)$ distance queries on random $\Delta$-regular graphs for an appropriately chosen $s$ (\cref{thm:random}).
The analysis  exploits the \emph{locally tree-like} property of random $\Delta$-regular graphs.
The proof of \cref{thm:random} consists of several technical novelties, based on a new concept of \emph{interesting vertices} (\cref{def:interesting}). See \cref{sec:random}.
\begin{theorem}
\label{thm:random}
Consider a uniformly random  $\Delta$-regular graph with $\Delta=O(1)$.
Let $s=\log^2 n$.
In the distance query model, \algoname (\cref{algo:main}) is a reconstruction algorithm using $\tilde O(n)$ queries in expectation.
\end{theorem}

We extend \algoname and its analysis to reconstruct random $\Delta$-regular graphs in the all-distances query model (\cref{cor:all-distances}), in the betweenness query model (\cref{cor:betweenness}), as well as in the parallel setting (\cref{cor:parallel}).
These extensions are based on the observation that the set $\hat E$ constructed in \algoname equals the edge set $E$ with high probability (\cref{lem:single-iteration}),\footnote{This property (i.e., $\hat E=E$ with high probability) does not hold on general graphs of bounded degree.}
see \cref{sec:variant}.

\begin{corollary}
\label{cor:all-distances}
Consider a uniformly random  $\Delta$-regular graph with $\Delta=O(1)$.
In the all-distances query model, there is a reconstruction algorithm using $O(\log^2 n)$ queries in expectation.
\end{corollary}

\begin{corollary}
\label{cor:betweenness}
Consider a uniformly random  $\Delta$-regular graph with $\Delta=O(1)$.
In the betweenness query model, there is a reconstruction algorithm using $\tilde O(n)$ queries in expectation.
\end{corollary}

\begin{corollary}
\label{cor:parallel}
Consider a uniformly random  $\Delta$-regular graph with $\Delta=O(1)$.
In the parallel setting of the distance query model, there is a reconstruction algorithm using $1+o(1)$ rounds and $\tilde O(n)$ queries in expectation.
\end{corollary}

We further extend the analysis of \algoname to study the metric dimension of random $\Delta$-regular graphs (\cref{cor:metric_dimention}), by showing (in \cref{lem:landmarks}) that a random subset of $\log^2 n$ vertices is almost surely a \emph{resolving set} (\cref{def:resolving}) for those graphs, see \cref{sec:metric_dimension}.

\begin{corollary}
\label{cor:metric_dimention}
Consider a uniformly random  $\Delta$-regular graph with $\Delta=O(1)$.
With probability $1-o(1)$, the metric dimension of the graph is at most $\log^2 n$.
\end{corollary}

With extra work, the parameter $s=\log^2 n$  in \cref{thm:random} can be reduced to $\log n\cdot (\log \log n)^{2+\epsilon}$, for any $\epsilon>0$, see \apref{remark:extra-work}.
As a consequence, the query complexity in the all-distances query model (\cref{cor:all-distances}) and the upper bound on the metric dimension (\cref{cor:metric_dimention}) can both be improved to $O(\log n\cdot (\log \log n)^{2+\epsilon})$.

\subsubsection{Bounded-Degree Graphs}
On general graphs of bounded degree, \algoname (\cref{algo:main}) has subquadratic query complexity and is highly parallelizable (\cref{thm:general}), see \cref{sec:general}.

\begin{theorem}
\label{thm:general}
Consider a general  graph of bounded degree $\Delta=O(\polylog n)$.
Let $s=n^{2/3}$.
In the distance query model, \algoname (\cref{algo:main}) is a reconstruction algorithm using $\tilde O(n^{5/3})$ queries in expectation.
In addition,  \algoname can be parallelized using 2 rounds.
\end{theorem}

We note that the \valgoname algorithm\footnote{Algorithm~3 in \cite{kannan2018graph}.} from \cite{kannan2018graph} also reconstructs graphs of bounded degree in the distance query model.
How does \algoname compare to \valgoname?
In terms of query complexity, on general graphs of bounded degree, \algoname uses $\tilde O(n^{5/3})$ queries, so is not as good as \valgoname using $\tilde O(n^{3/2})$ queries; on random $\Delta$-regular graphs, \algoname is more efficient than \valgoname: $\tilde O(n)$ versus $\tilde O(n^{3/2})$.
In terms of round complexity, \algoname can be parallelized using 2 rounds on general graphs of bounded degree, and even $1+o(1)$ rounds on random  $\Delta$-regular graphs; while \valgoname requires up to $3\log n$ rounds due to a multi-phase selection process for centers.\footnote{The number of rounds in \valgoname is implicit in the proof of Lemma~2.3 from \cite{kannan2018graph}.}
In terms of structure, \algoname is much simpler than \valgoname, which is based on multi-phase Voronoi-cell decomposition.

In worst case instances of graphs of bounded degree, the query complexity of \algoname is higher than linear.
For example, when the graph is a complete binary tree, \algoname would require $\Omega(n\sqrt{n})$ queries (the complexity of \algoname is minimized when $s$ is roughly $\sqrt{n}$).
Thus the open question from \cite{kannan2018graph} of whether general graphs of bounded degree can be reconstructed using $\tilde O(n)$ distance queries remains open and answering it positively would require further algorithmic ideas.

\section{Notations and Preliminary Analysis}
\label{sec:preliminary}
Let $G=(V,E)$ be a connected, undirected, and unweighted graph, where $V$ is the set of vertices such that $|V|=n$ and $E$ is the set of edges.
We say that $\{a,b\}\subseteq V$ is a \emph{vertex pair} if both $a$ and $b$ belong to $V$ such that $a\neq b$.
The \emph{distance} between a vertex pair $\{a,b\}\subseteq V$, denoted by $\delta(a,b)$, is the number of edges on a shortest $a$-to-$b$ path.

\begin{definition}[Distinguishing]
\label{def:distinguish}
For a vertex pair $\{a,b\}\subseteq V$, we say that a vertex $u\in V$ \emph{distinguishes} $a$~and~$b$, or equivalently that $u$ is a \emph{distinguisher} of $\{a,b\}$, if $|\delta(u,a)-\delta(u,b)|>1$.
Let $D(a,b)\subseteq V$ denote  the set of vertices $u\in V$ distinguishing $a$ and $b$.
\end{definition}

Let $s\in[1,n]$ be an integer parameter.
The set $S$ constructed in \algoname consists of $s$ vertices selected uniformly and independently at random from $V$.

The set $\hat E$ constructed in \algoname consists of the vertex pairs $\{a,b\}\subseteq V$ such that $a$ and $b$ are not distinguished by any vertex in $S$, i.e., $D(a,b)\cap S= \emptyset$, or equivalently, $|\delta(u,a)-\delta(u,b)|\leq 1$ for all $u\in S$.
For any edge $(a,b)\in E$, it is easy to see that $|\delta(u,a)-\delta(u,b)|\leq 1$ for all $u\in V$, which implies that $\{a,b\}\in \hat E$.
Hence the following inclusion property.

\begin{fact}
\label{fact:hat-E}
$E\subseteq \hat E$.
\end{fact}

We show that \algoname is correct and we give a preliminary analysis on its query complexity as well as on its round complexity, in \cref{lem:query}.

\begin{lemma}
\label{lem:query}
The output of \algoname (\cref{algo:main}) equals the edge set $E$.
The number of distance queries in \algoname is $n\cdot s+|\hat E|$.
In addition, \algoname can be parallelized using 2 rounds.
\end{lemma}

\begin{proof}
The output of \algoname consists of the vertex pairs $\{a,b\}\in \hat E$ such that $\{a,b\}$ is an edge in $E$.
Since $E\subseteq \hat E$ (\cref{fact:hat-E}), the output of \algoname equals the edge set $E$.

Observe that the distance queries in \algoname are performed in two stages.
The number of distance queries in the first stage is $|V|\cdot |S|=n\cdot s$.
The number of distance queries in the second stage is $|\hat E|$.
Thus the query complexity of \algoname is $n\cdot s+|\hat E|$.
The distance queries in each of the two stages can be performed in parallel, so \algoname can be parallelized using 2 rounds.
\end{proof}

From \cref{lem:query}, in order to further study the query complexity of \algoname, it suffices to analyze $|\hat E|$, which equals $|E|+|\hat E\setminus E|$ according to \cref{fact:hat-E}.
Since $|E|\leq \Delta n$ in a graph of bounded degree $\Delta$, our focus in the subsequent analysis is $|\hat E\setminus E|$.

\begin{lemma}
\label{lem:B}
Let $s=\omega(\log n)$ be an integer parameter.
Let $B$ be the set of vertex pairs $\{a,b\}\subseteq V$ such that $\delta(a,b)\geq 2$ and $|D(a,b)|\leq 3n\cdot (\log n)/s$.
We have $\expect{S}{|\hat E\setminus E|}\leq |B|+o(1)$.
\end{lemma}

\begin{proof}
Denote $Z$ as the set $\hat E\setminus E$.
Observe that $|Z|\leq |B|+|Z\setminus B|$.
Since $B$ is independent of the random set $S$, we have $\expect{S}{|Z|}\leq |B|+\expect{S}{|Z\setminus B|}$.
It suffices to show that $\expect{S}{|Z\setminus B|}=o(1)$.

We claim that for any vertex pair $\{a,b\}\subseteq V$ such that $\{a,b\}\notin B$, the probability that $\{a,b\}\in Z$ is $o(n^{-2})$.
To see this, fix a vertex pair $\{a,b\}\notin B$.
By definition of $B$, either $\delta(a,b)=1$, or $|D(a,b)|>3n\cdot (\log n)/s$.
In the first case, $\{a,b\}\notin Z$ since $Z$ does not contain any edge of $E$.
In the second case, the event $\{a,b\}\in Z$ would imply that $\{a,b\}\in \hat E$, hence $D(a,b)\cap S= \emptyset$. Therefore,
\begin{align*}
&\proba{S}{\{a,b\}\in Z\mid\{a,b\}\notin B}\\
\leq & \proba{S}{D(a,b)\cap S= \emptyset\mid\{a,b\}\notin B}\\
< &\left(1-\frac{3n\cdot (\log n)/s}{n}\right)^s\\
=& o(n^{-2}),
\end{align*}
where the second inequality follows since $|D(a,b)|>3n\cdot (\log n)/s$ and the set $S$ consists of $s$ vertices selected uniformly and independently at random, and the last step follows since $s=\omega(\log n)$.

There are at most $n(n-1)/2$ vertex pairs $\{a,b\}\notin B$.
By the linearity of expectation, the expected number of vertex pairs $\{a,b\}\notin B$ such that $\{a,b\}\in Z$ is at most $o(n^{-2})\cdot n(n-1)/2=o(1)$, so $\expect{S}{|Z\setminus B|}=o(1)$.
Therefore, $\expect{S}{|Z|}\leq |B|+\expect{S}{|Z\setminus B|}= |B|+o(1)$.
\end{proof} 
\section{Reconstruction of Random Regular Graphs (Proof of \cref{thm:random})}
\label{sec:random}
In this section, we analyze \algoname (\cref{algo:main}) on random  $\Delta$-regular graphs in the distance query model.
We assume that $\Delta\geq 2$ and that $\Delta n$ is even since otherwise those graphs do not exist.

We bound the expectation of $|\hat E\setminus E|$ on random  $\Delta$-regular graphs, in \cref{lem:random}.

\begin{lemma}
\label{lem:random}
Let $G$ be a uniformly random $\Delta$-regular graph with $\Delta=O(1)$.
Let $s=\log^2 n$.
Let $S\subseteq V$ be a set of $s$ vertices selected uniformly and independently at random from $V$.
We have $\expect{G,S}{|\hat E\setminus E|}=o(1)$.
\end{lemma}

\begin{proof}[Proof of \cref{thm:random} using \cref{lem:random}]
By \cref{lem:query}, \algoname is a reconstruction algorithm using $n\cdot s+|\hat E|=n\cdot\log^2 n+|\hat E|$ distance queries.
From \cref{fact:hat-E}, $|\hat E|=|E|+|\hat E\setminus E|$.
Since $G$ is $\Delta$-regular, $|E|=\Delta n/2$.
By \cref{lem:random}, $\expect{G,S}{|\hat E\setminus E|}=o(1)$.
Therefore, the expected number of distance queries in \algoname is $n\cdot\log^2 n+\Delta n/2+o(1)$, which is $\tilde O(n)$ since $\Delta=O(1)$.
\end{proof}

It remains to prove \cref{lem:random} in the rest of this section.

\subsection{Configuration Model and the Structural Lemma}
We consider  a random $\Delta$-regular graph generated according to the \emph{configuration model}~\cite{bollobas1980probabilistic,wormald1999models}.
Given a partition of a set of $\Delta n$ points  into $n$ cells $v_1, v_2, \dots, v_n$ of $\Delta$ points, a \emph{configuration} is a perfect matching of the points into $\Delta n/2$ pairs. It corresponds to a (not necessarily connected) \emph{multigraph}  $G'$  in which the cells are regarded as vertices and the pairs as edges: a pair of points $\{x, y\}$ in the configuration corresponds to an edge $(v_i, v_j)$ of $G'$ where $x\in v_i$ and $y\in v_j$.
Since each  $\Delta$-regular graph has exactly  $(\Delta!)^n$ corresponding configurations, a  $\Delta$-regular graph can be generated uniformly at random by rejection sampling: choose a configuration uniformly at random,\footnote{To generate a random configuration, the points in a pair can be chosen sequentially:  the first point can be selected using any rule, as long as the second point in that pair is chosen uniformly from the remaining points.} and reject the result if the corresponding multigraph $G'$ is not simple or not connected.
The configuration model enables us to show properties of a random $\Delta$-regular graph by analyzing a multigraph $G'$ corresponding to a random configuration.

Based on the configuration model, we are ready to state the following \emph{Structural Lemma}, which is central in our analysis.
\begin{lemma}[Structural Lemma]
\label{lem:Fvw}
Let $\Delta=O(1)$ be such that $\Delta\geq 3$.
Let $G'$ be a multigraph corresponding to a uniformly random configuration.
Let $\{v,w\}$ be a vertex pair in $G'$ such that $\delta(v,w)\geq 2$.
With probability $1-o(n^{-2})$, we have $|D(v,w)|>3n/\log n$.
\end{lemma}

In \cref{sec:proof:Fvw}, we prove the Structural Lemma, and in \cref{sec:proof-random}, we show \cref{lem:random} using the Structural Lemma.

\subsection{Proof of the Structural Lemma (\cref{lem:Fvw})}
\label{sec:proof:Fvw}
Let $G'$ be a multigraph corresponding to a uniformly random configuration, and let $V$ be the vertex set of $G'$.
Let $\{v,w\}\subseteq V$ be a vertex pair such that $\delta(v,w)\geq 2$.
For a vertex $x\in V$, denote $\ell(x)\in \mathbb{Z}$ as the distance in $G'$ between $x$ and the vertex pair $\{v,w\}$, i.e., $\ell(x)=\min(\delta(x,v),\delta(x,w))$.
For any integer $k\geq 0$, denote $U_k\subseteq V$ as the set of vertices $x\in V$ such that $\ell(x)=k$.
Denote $U_{\leq k}=\bigcup_{j\leq k}U_j$.

To construct the multigraph $G'$ from a random configuration, we borrow the approach from~\cite{bollobas1982distinguishing}, which proceeds in $n$ phases to construct the edges in $G'$, exploring vertices $x\in V$ in non-decreasing order of $\ell(x)$.
We start at the vertices of $U_0=\{v,w\}$.
Initially, i.e., in the $0$-th phase, we construct all the edges incident to $v$ or incident to $w$.
Suppose at the beginning of the $k$-th phase, for each $k\in[1,n-1]$, we have constructed all the edges with at least one endpoint belonging to $U_{\leq k-1}$.
During the $k$-th phase, we construct the edges incident to the vertices in $U_k$ one by one, till the degree of all the vertices in $U_k$ reaches $\Delta$.
The ordering of the edge construction within the same phase is arbitrary.
Let $G'$ be the resulting multigraph in the end of the construction.\footnote{When a multigraph corresponding to a random configuration is not connected, the resulting $G'$ consists of the union of the components of $v$ and of $w$, respectively, in that multigraph.
Note that any vertex $x\in V$ outside those two components cannot distinguish $v$ and $w$ (i.e., $x\notin D(v,w)$), thus $x$ is irrelevant to $|D(v,w)|$ in the statement of \cref{lem:Fvw}.}
The ordering of the edges in $G'$ is defined according to the above edge construction.

An edge $(a,b)$ in $G'$ is \emph{indispensable} if it explores either the vertex $a$ or the vertex $b$ for the first time in the edge construction.
In the first case, $b$ is the \emph{predecessor} of $a$; and in the second case, $a$ is the \emph{predecessor} of $b$.
An edge is \emph{dispensable} if it is not indispensable, in other words, if each of its endpoints either belongs to $\{v,w\}$ or is an endpoint of an edge constructed previously.

\begin{fact}
\label{fact:dispensable}
Neither $v$ or $w$ has a predecessor.
For any vertex in $V$, its predecessor, if exists, is unique.
If vertex $a$ is the predecessor of vertex $b$, then $\ell(b)=\ell(a)+1$.
\end{fact}

We introduce the concept of \emph{interesting vertices}, which is a key idea in the analysis.

\begin{definition}[Interesting Vertices]
\label{def:interesting}
A vertex $x\in V$ is \emph{$v$-interesting} if, for all vertices $z\in V\setminus\{v\}$ with $\delta(v,z)+\delta(z,x)=\delta(v,x)$, the edges incident to $z$ are indispensable.
Similarly, a vertex $x\in V$ is \emph{$w$-interesting} if, for all vertices $z\in V\setminus\{w\}$ with $\delta(w,z)+\delta(z,x)=\delta(w,x)$, the edges incident to $z$ are indispensable.
\end{definition}
For any finite integer $k\geq 1$, let $I_k(v)\subseteq V$ denote the set of $v$-interesting vertices $x\in V$ such that $\delta(v,x)=k$, and let $I_k(w)\subseteq V$ denote the set of $w$-interesting vertices $x\in V$ such that $\delta(w,x)=k$.

We show in \cref{lem:interesting-distinguish} that interesting vertices distinguish the vertex pair $\{v,w\}$, and we provide a lower bound on the number of interesting vertices in \cref{lem:number-interesting}.
These two lemmas are main technical novelties in our work.
Their proofs are in \cref{sec:proof-interesting-distinguish,sec:proof:number-interesting}, respectively.

\begin{lemma}
\label{lem:interesting-distinguish}
For any finite integer $k\geq 1$, we have $I_k(v)\cup I_k(w)\subseteq D(v,w)$.
\end{lemma}

\begin{lemma}
\label{lem:number-interesting}
Let $\Delta=O(1)$ be such that $\Delta\geq 3$.
Let $k$ be any positive integer such that $k\leq \lceil \log_{\Delta-1} (3n/\log n) \rceil+2$.
With probability $1-o(n^{-2})$, we have $|I_k(v)\cup I_k(w)|> (\Delta-2-o(1))(\Delta-1)^{k-1}.$
\end{lemma}

The Structural Lemma (\cref{lem:Fvw}) follows easily from \cref{lem:interesting-distinguish,lem:number-interesting}, see \cref{sec:proof-structral-lemma}.

\subsubsection{Proof of \cref{lem:interesting-distinguish}}
\label{sec:proof-interesting-distinguish}
Fix a finite integer $k\geq 1$.
From the symmetry of $v$ and $w$, it suffices to prove $I_k(v)\subseteq D(v,w)$.

Let $x$ be any vertex in $I_k(v)$.
By definition, $x$ is $v$-interesting and $\delta(v,x)=k$.
Let $a_0 =v,a_1,\dots,a_k=x$ be any shortest $v$-to-$x$ path.
For any vertex $a_i$ with $i\in[1,k]$, the edges incident to $a_i$ are indispensable according to \cref{def:interesting}.

We claim that, for any $i\in[1,k]$, $a_{i-1}$ is the predecessor of $a_i$, and in addition, $\ell(a_i)=i$.
The proof is by induction.
First, consider the case when $i=1$.
The edge $(a_0,a_1)$ is incident to the vertex $a_1$, so is indispensable.
Thus either $a_0$ is the predecessor of $a_1$, or $a_1$ is the predecessor of $a_0$.
Since $a_0$ $(=v)$ has no predecessor (\cref{fact:dispensable}), $a_1$ cannot be the predecessor of $a_0$, so $a_0$ is the predecessor of $a_1$.
Again using \cref{fact:dispensable}, we have $\ell(a_1)=\ell(a_0)+1$.
Since $\ell(a_0)=\ell(v)=0$, we have $\ell(a_1)=1$.
Next, consider the case when $i\geq 2$, and assume that the claim  holds already for $1,\dots,i-1$.
The edge $(a_{i-1},a_i)$ is incident to the vertex $a_i$, so is indispensable.
Thus either $a_{i-1}$ is the predecessor of $a_i$, or $a_i$ is the predecessor of $a_{i-1}$.
By induction, $a_{i-2}$ is the predecessor of $a_{i-1}$.
Since the predecessor of $a_{i-1}$ is unique (\cref{fact:dispensable}), $a_i$ cannot be the predecessor of $a_{i-1}$, so $a_{i-1}$ is the predecessor of $a_i$.
Again using \cref{fact:dispensable}, we have $\ell(a_i)=\ell(a_{i-1})+1$.
Since $\ell(a_{i-1})=i-1$ by induction, we have $\ell(a_i)=i$.

In order to show that $x\in D(v,w)$, we prove in the following that $\delta(w,x)\geq k+2$.
Indeed, since $\delta(v,x)=k$, the event $\delta(w,x)\geq k+2$ implies that $x\in D(v,w)$ by \cref{def:distinguish}.\footnote{When $\delta(w,x)$ is infinite (i.e., $w$ and $x$ are not connected in $G'$), it is trivial that $x\in D(v,w)$, since $\delta(v,x)$ is finite. Therefore, it suffices to consider the case when $\delta(w,x)$ is finite in the rest of the proof.}

Let $b_0 =w,b_1,\dots,b_{k'}=x$ be any shortest $w$-to-$x$ path, for some integer $k'$.
See \cref{fig:dispensable}.
Let $i^*\in [0,k]$ be the largest integer such that $a_{k-j}=b_{k'-j}$ for all $j\in[0,i^*]$.
Let $z$ denote the vertex $a_{k-i^*}$, which equals $b_{k'-i^*}$.
If $i^*=k$, the $v$-to-$x$ path $a_0,a_1,\dots,a_k$ is a subpath of the $w$-to-$x$ path $b_0,b_1,\dots,b_{k'}$.
Since $\delta(w,v)\geq 2$, we have $\delta(w,x)=\delta(w,v)+\delta(v,x)\geq 2+k$, which implies that $x\in D(v,w)$.
From now on, it suffices to consider the case when $i^*<k$.

\begin{figure}[t]
\centering
\includegraphics[scale=0.8]{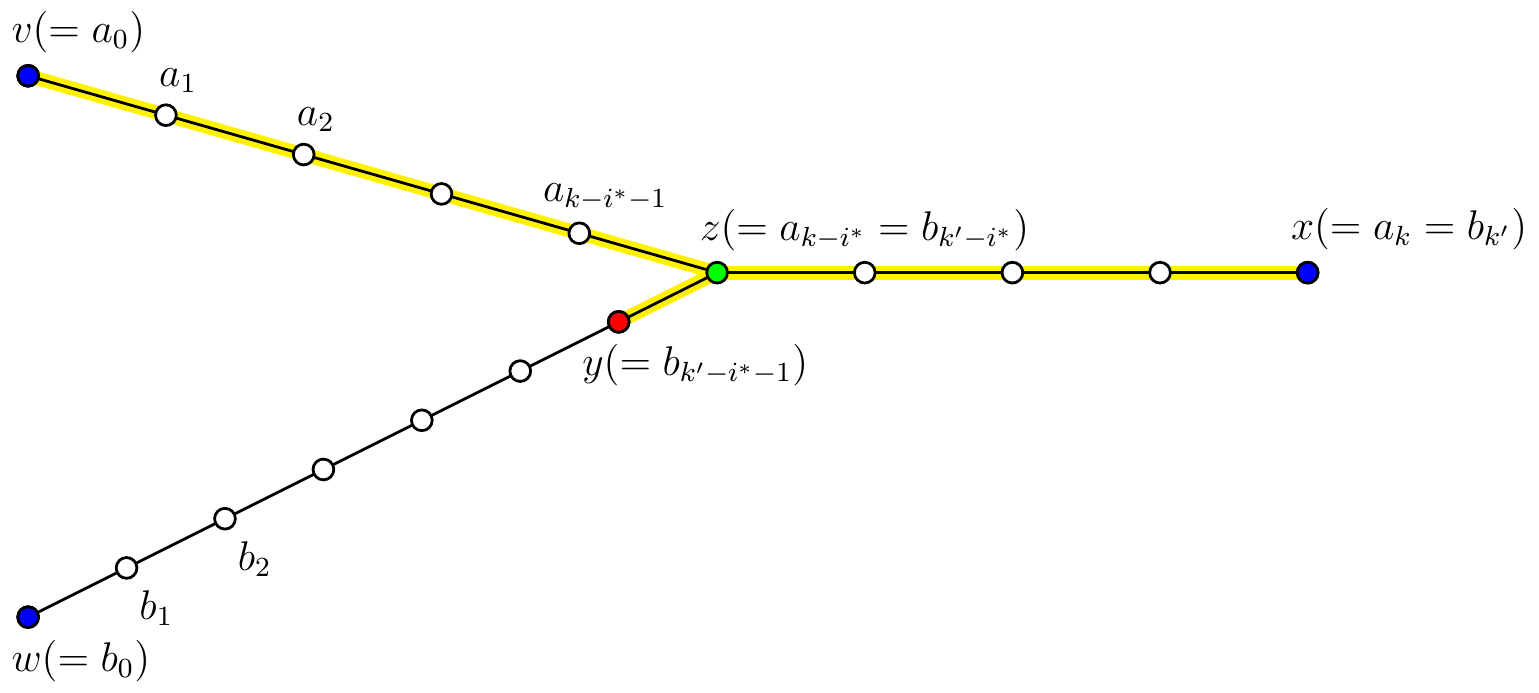}
\caption{
$a_0,a_1,\dots,a_k$ is a shortest $v$-to-$x$ path, and $b_0,b_1,\dots,b_{k'}$ is
a shortest $w$-to-$x$ path.
The vertex $z$ represents the branching point of these two paths.
Since the vertex $x$ is $v$-interesting, the highlighted edges are indispensable.
}
\label{fig:dispensable}
\end{figure}

Let $y$ denote the vertex $b_{k'-i^*-1}$.
Since $y$ is on a shortest $w$-to-$x$ path, we have
\begin{equation}
\label{eqn-wx}
\delta(w,x)=\delta(w,y)+\delta(y,x)=\delta(w,y)+(i^*+1)\geq \ell(y)+(i^*+1),
\end{equation}
where the inequality follows from the definition of $\ell(y)$.
It remains to analyze the value of $\ell(y)$.

The edge $(z,y)$ is incident to the vertex $z$ $(=a_{k-i^*})$, so is indispensable.
Thus either $y$ is the predecessor of $z$, or $z$ is the predecessor of $y$.
From the previous claim, $a_{k-i^*-1}$ is the predecessor of $z$.
Since the predecessor of $z$ is unique (\cref{fact:dispensable}) and $y\neq a_{k-i^*-1}$ (by definition of $i^*$), $y$ cannot be the predecessor of $z$, so $z$ is the predecessor of $y$.
Again by \cref{fact:dispensable}, $\ell(y)=\ell(z)+1$.
Since $\ell(z)=\ell(a_{k-i^*})=k-i^*$ by the previous claim, we have $\ell(y)=k-i^*+1$.
We conclude from \cref{eqn-wx} that \[\delta(w,x)\geq (k-i^*+1)+(i^*+1)=k+2,\] which implies that $x\in D(v,w)$.

We proved that $I_k(v)\subseteq D(v,w)$.
Similarly, $I_k(w)\subseteq D(v,w)$.
Therefore, $I_k(v)\cup I_k(w)\subseteq D(v,w)$.

We complete the proof of \cref{lem:interesting-distinguish}.

\subsubsection{Proof of \cref{lem:number-interesting}}
\label{sec:proof:number-interesting}

To begin with, we show that there are relatively few dispensable edges within a neighborhood of $\{v,w\}$.
This property, also called the \emph{locally tree-like} property, was previously exploited by Bollob{\'a}s~\cite{bollobas1982distinguishing} for three levels of neighborhoods on random $\Delta$-regular graphs  in the context of automorphisms of those graphs.
In \cref{lem:dispensable}, we extend the analysis from \cite{bollobas1982distinguishing} to show the locally tree-like property for $M=\lceil\log\log n\rceil$ levels of neighborhoods.

\begin{lemma}
\label{lem:dispensable}
Let $M=\lceil\log\log n\rceil$.
We can construct two non-decreasing sequences $\{g_i\}_{1\leq i\leq M}$ and $\{L_i\}_{1\leq i\leq M}$, such that all of the following properties hold when $n$ is large enough:
\begin{enumerate}
\item \label{prop:g_i} $g_1=3$; and for any $i\in [2,M]$, $g_{i}=o\left((\Delta-1)^{L_{i-1}}/M\right)$;
\item \label{prop:ell_M} $L_M\geq \lceil \log_{\Delta-1} (3n/\log n) \rceil+2$;
\item \label{prop:dispensable} With probability $1-o(n^{-2})$, for all $i\in [1,M]$, strictly less than $g_i$ edges are dispensable among the edges incident to vertices in $U_{\leq L_i}$.
\end{enumerate}
\end{lemma}

\begin{proof}[Proof of \cref{lem:dispensable}]
First, we define two sequences $\{g_i\}_{1\leq i\leq M}$ and $\{f_i\}_{1\leq i\leq M}$ as follows: $g_1=3$, $f_1=\left\lceil n^{1/8}\right\rceil$, and for each $i\in [2,M]$, let
\[g_i=\left\lceil n^{1-7/2^{i+1}}/(\log n)^{1/2}\right\rceil,\]
\[f_i=\left\lceil n^{1-7/2^{i+2}}/(\log n)^{1/3}\right\rceil.\]
Next, we define the sequence $\{L_i\}_{1\leq i\leq M}$ as follows: for each $i\in[1,M]$, let \[L_i=\left\lceil\log_{\Delta-1} f_i\right\rceil-6.\]
It is easy to see that all of the three sequences $\{g_i\}_{1\leq i\leq M}$, $\{f_i\}_{1\leq i\leq M}$, and $\{L_i\}_{1\leq i\leq M}$ are non-decreasing.

To show Property~\ref{prop:g_i} of the statement, observe that for any $i\in[2,M]$,
\[g_i\cdot M=\left\lceil n^{1-7/2^{i+1}}/(\log n)^{1/2}\right\rceil\cdot \lceil\log\log n\rceil=o\left( n^{1-7/2^{i+1}}/(\log n)^{1/3}\right).\]
Thus $g_i\cdot M=o(f_{i-1})$ by definition of $f_{i-1}$.
From the definition of $L_{i-1}$ and the fact that $\Delta=O(1)$, we have $f_{i-1}=\Theta\left((\Delta-1)^{L_{i-1}}\right)$.
Therefore, $g_i\cdot M =o\left((\Delta-1)^{L_{i-1}}\right)$, hence $g_i =o\left((\Delta-1)^{L_{i-1}}/M\right)$.

To show Property~\ref{prop:ell_M}  of the statement, observe that
\[f_M\geq  n^{1-7/2^{(\log\log n)+2}}/(\log n)^{1/3}=2^{-7/4}\cdot n/(\log n)^{1/3}>(\Delta-1)^8\cdot 3n/\log n,\]
where the last inequality follows since $\Delta=O(1)$ and $n$ is large enough. Therefore,
$L_M=\lceil\log_{\Delta-1} f_M\rceil-6\geq \lceil \log_{\Delta-1} (3n/\log n) \rceil+2$.

It remains to show Property~\ref{prop:dispensable} of the statement.
Consider any integer $i\in[1,M]$.
Since the graph is $\Delta$-regular, the number of vertices in $U_{\leq L_i}$ is at most \[2+2\Delta\cdot\sum_{j=0}^{L_i-1}(\Delta-1)^j=2+\frac{2\Delta(\Delta-1)^{L_i}-2\Delta}{\Delta-2}<\frac{2\Delta(\Delta-1)^{L_i}}{\Delta-2}.\]
Let $n_i$ be the number of edges incident to vertices in $U_{\leq L_i}$.
Since each vertex is incident to $\Delta$ edges, we have $n_i< \frac{2\Delta^2(\Delta-1)^{L_i}}{\Delta-2}$.
Since $L_i=\left\lceil\log_{\Delta-1} f_i\right\rceil-6<\left(\log_{\Delta-1} f_i\right)-5$, we have
$n_i<\frac{2\Delta^2}{(\Delta-2)(\Delta-1)^5} \cdot f_i$, which is less than $f_i$ since $\Delta\geq 3$.

In order to bound the number of dispensable edges incident to vertices in $U_{\leq L_i}$, it suffices to bound the number of dispensable edges
among the first $f_i$ edges in the ordering of edge construction.

For any integer $t\in [1,\Delta n/2]$, denote $p(t)$ as the probability that the $t$-th edge in the construction is dispensable.
We use the argument of Bollob{\'a}s~\cite{bollobas1982distinguishing} to bound $p(t)$ as follows.
Before constructing the $t$-th edge, the previously constructed $t-1$ edges are incident to at most $t+1$ vertices.
For each of these $t+1$ vertices, at most $\Delta-1$ incident edges are not yet constructed.
Thus $p(t)\leq \frac{(\Delta-1) (t+1)}{\Delta n-2(t-1)}$, which is less than $\frac{2t}{n}$ as soon as $t=o(n)$.

From the definition of $f_i$, we have $f_i\leq n/(\log n)^{1/3}=o(n)$, thus $p(f_i)<\frac{2f_i}{n}$.
The probability that there exist $g_i$ dispensable edges among the first $f_i$ edges is at most
\[\binom{f_i}{g_i}\cdot \left(\frac{2f_i}{n}\right)^{g_i}<\left(\frac{e\cdot f_i}{g_i}\right)^{g_i}\cdot \left(\frac{2f_i}{n}\right)^{g_i},\]
where the inequality follows from Stirling's formula.
When $i=1$, we have
\[\left(\frac{e\cdot f_1}{g_1}\right)^{g_1}\cdot \left(\frac{2f_1}{n}\right)^{g_1}= \left(\frac{2e\cdot \left(\lceil n^{1/8}\rceil\right)^2}{3n}\right)^3= o(n^{-17/8}),\] and when $i\geq 2$, we have
\[\left(\frac{e\cdot f_i}{g_i}\right)^{g_i}\cdot \left(\frac{2f_i}{n}\right)^{g_i}=O\left(\left(\frac{2e}{(\log n)^{1/6}}\right)^{g_i}\right)=o(n^{-17/8}),\] by definition of $g_i$ and $f_i$ and by observing that $g_i\geq n^{1/8}/(\log n)^{1/2}$ for any $i\geq 2$.

Thus for any $i\in [1,M]$, with probability $1-o(n^{-17/8})$, strictly less than $g_i$ edges are dispensable among the first $f_i$ edges, hence strictly less than $g_i$ edges are dispensable among the edges incident to vertices in $U_{\leq L_i}$.

Therefore, with probability $1-o(M\cdot n^{-17/8})=1-o(n^{-2})$, for all $i\in [1,M]$, strictly less than $g_i$ edges are dispensable among the edges incident to vertices in $U_{\leq L_i}$.
This completes the proof for Property~\ref{prop:dispensable}  of the statement.
\end{proof}

We condition on the occurrence of the high probability event in Property~\ref{prop:dispensable} of \cref{lem:dispensable}.
Let $\mathcal{E}$ denote this event.

We say that a dispensable edge is \emph{trivial} if it is incident to $v$ or incident to $w$, and \emph{non-trivial} otherwise.
Let $E_0$ be the set of trivial dispensable edges.
Let $E_1$ be the set of non-trivial dispensable edges that are incident to vertices in $U_{\leq L_1}$.
The event $\mathcal{E}$ implies that strictly less than $g_1(=3)$ edges are dispensable among the edges incident to vertices in $U_{\leq L_1}$.
Hence $|E_0|+|E_1|\leq 2$.

Let $F_0\subseteq U_1$ be the set of vertices $u\in U_1$ such that $u$ is not incident to any trivial dispensable edge.
We claim that $|F_0|\geq 2\Delta-2|E_0|$.
If $E_0=\emptyset$, it is clear that $|F_0|=2\Delta$.
If $E_0\neq \emptyset$, there are three cases for each trivial dispensable edge in $E_0$: (1) a self-loop at $v$ or at $w$, (2) a parallel edge incident to $v$ or incident to $w$, and (3) an edge $(v,u)$ when $u$ is a neighbor of $w$, or an edge $(w,u)$ when $u$ is a neighbor of $v$.
In all the three cases, the existence of each trivial dispensable edge in $E_0$ decreases the size of $F_0$ by at most 2.
Hence $|F_0|\geq 2\Delta-2|E_0|$.

For each $u\in F_0$, define \[T(u)=\{x\in U_{\leq L_1}\mid \ell(x)=\delta(x,u)+1\}.\]
Let $F\subseteq F_0$ be the set of vertices $u\in F_0$ such that $T(u)$ contains no vertex incident to a dispensable edge in $E_1$.
Since each dispensable edge in $E_1$ is incident to two vertices, we have \[|F|\geq |F_0|-2|E_1|\geq 2\Delta-2|E_0|-2|E_1|\geq 2(\Delta-2).\]
Since $F\subseteq F_0\subseteq U_1$, one of $v$ and $w$ has at least $|F|/2\geq \Delta-2$ neighbors in $F$.

Without loss of generality, we assume that $v$ has at least $\Delta-2$ neighbors in $F$.
We show that, under this assumption, $|I_k(v)|\geq (\Delta-2-o(1))(\Delta-1)^{k-1}$.

Our proof proceeds in increasing order on $k\geq 1$.

First, consider any integer $k\in[1,L_1]$.
Let $u$ be any neighbor of $v$ in $F$.
Since $T(u)$ contains no vertex incident to a dispensable edge, $T(u)$ corresponds to a complete $(\Delta-1)$-ary tree.
Consider any vertex $x\in T(u)$ such that $\delta(v,x)=k$.
Any vertex $z\in V\setminus\{v\}$ such that $\delta(v,z)+\delta(z,x)=\delta(v,x)$ belongs to the (unique) shortest $x$-to-$u$ path.
Since the shortest $x$-to-$u$ path is completely within $T(u)$, we have $z\in T(u)$, thus the edges incident to $z$ are indispensable.
Hence $x$ is $v$-interesting according to \cref{def:interesting}.
Since $\delta(v,x)=k$, we have $x\in I_k(v)$.
There are at least $\Delta-2$ choices of $u$, and for a fixed $u$, there are $(\Delta-1)^{k-1}$ choices of $x$.
Therefore, the size of $I_k(v)$ is at least $(\Delta-2)(\Delta-1)^{k-1}$.

Next, consider any integer $k\in [L_1+1,L_2]$.
For any vertex $x\in I_{L_1}(v)$, define
\[T'(x)=\{y\in U_{\leq L_2}\mid \ell(y)=\delta(y,x)+L_1\}.\]
Let $F'\subseteq I_{L_1}(v)$ be the set of vertices $x\in I_{L_1}(v)$ such that $T'(x)$ contains no vertex incident to a dispensable edge.
The event $\mathcal{E}$ implies that strictly less than $g_2$ dispensable edges are incident to vertices in $U_{\leq L_2}$.
Since each dispensable edge is incident to two vertices, we have $|F'|> |I_{L_1}(v)|-2g_2$.
Let $x$ be any vertex in $F'$.
Since $T'(x)$ contains no vertex incident to a dispensable edge, $T'(x)$ corresponds to a complete $(\Delta-1)$-ary tree.
Consider any vertex $y\in T'(x)$ such that $\delta(v,y)=k$.
Any vertex $z\in V\setminus\{v\}$ such that $\delta(v,z)+\delta(z,y)=\delta(v,y)$ belongs either to the (unique) shortest $x$-to-$v$ path or to the (unique) shortest $y$-to-$x$ path.
In the first case, since $x$ is $v$-interesting, the edges incident to $z$ are indispensable by \cref{def:interesting}.
In the second case, since the shortest $y$-to-$x$ path is completely within $T'(x)$, we have $z\in T'(x)$, thus the edges incident to $z$ are indispensable.
Hence $y$ is $v$-interesting according to  \cref{def:interesting}.
Since $\delta(v,y)=k$, we have $y\in I_k(v)$.
There are $|F'|> |I_{L_1}(v)|-2g_2$ choices of $x$, and for a fixed $x$, there are $(\Delta-1)^{k-L_1}$ choices of $y$.
Therefore,
\begin{align*}
|I_k(v)|&> (|I_{L_1}(v)|-2g_2)\cdot (\Delta-1)^{k-L_1}\\
&\geq ((\Delta-2)(\Delta-1)^{L_1-1}-2g_2)\cdot (\Delta-1)^{k-L_1}\\
&= (\Delta-2-o(1/M))(\Delta-1)^{k-1},
\end{align*}
where the equality follows because $g_2=o((\Delta-1)^{L_1}/M)$ from \cref{lem:dispensable} and since $\Delta=O(1)$.

We move on to larger values of $k$.
Let $i$ be any integer in $[3,M]$.
From \cref{lem:dispensable}, strictly less than $g_i$ edges are dispensable among the edges incident to vertices in $U_{\leq L_i}$ and that $g_{i}=o\left((\Delta-1)^{L_{i-1}}/M\right)$.
For any integer $k\in[L_{i-1}+1,L_{i}]$, by extending the previous argument, we have
\[|I_k(v)|> (\Delta-2-i\cdot o(1/M))(\Delta-1)^{k-1}=(\Delta-2-o(1))(\Delta-1)^{k-1}.\]

We conclude that for any $k\in[1,L_M]$, we have $|I_k(v)|\geq (\Delta-2-o(1))(\Delta-1)^{k-1}$.
In the other case that $w$ has at least $\Delta-2$ neighbors in $F$, similarly, we have $|I_k(w)|\geq (\Delta-2-o(1))(\Delta-1)^{k-1}.$
Hence $|I_k(v)\cup I_k(w)|\geq (\Delta-2-o(1))(\Delta-1)^{k-1}$.
The event $\mathcal{E}$, on which the above analysis is conditioned, occurs with probability $1-o(n^{-2})$ according to \cref{lem:dispensable}.
Therefore, with probability $1-o(n^{-2})$, we have
\[|I_k(v)\cup I_k(w)|\geq (\Delta-2-o(1))(\Delta-1)^{k-1}, \text{ for any } k\in[1,L_M].\]
Again by \cref{lem:dispensable}, we have $L_M\geq \lceil \log_{\Delta-1} (3n/\log n) \rceil+2$.
Thus the above inequality holds for any positive integer $k\leq \lceil \log_{\Delta-1} (3n/\log n) \rceil+2$.

We complete the proof of \cref{lem:number-interesting}.

\subsubsection{Proof of the Structural Lemma (\cref{lem:Fvw}) using \cref{lem:interesting-distinguish,lem:number-interesting}}
\label{sec:proof-structral-lemma}
We set $k=\lceil \log_{\Delta-1} (3n/\log n) \rceil+2$.
By \cref{lem:interesting-distinguish}, $|D(v,w)|\geq |I_{k}(v)\cup I_{k}(w)|.$
By \cref{lem:number-interesting}, with probability $1-o(n^{-2})$, we have  \[|I_{k}(v)\cup I_{k}(w)|>(\Delta-2-o(1))(\Delta-1)^{k-1}\geq (\Delta-2-o(1))(\Delta-1)\cdot(3n/\log n),\]
where the last inequality follows from the definition of $k$.
Since $\Delta\geq 3$, we have $(\Delta-2-o(1))(\Delta-1)>1$.
Thus with probability $1-o(n^{-2})$, we have $|I_{k}(v)\cup I_{k}(w)|>3n/\log n$, which implies that $|D(v,w)|>3n/\log n$.

We complete the proof of \cref{lem:Fvw}.

\subsection{Proof of \cref{lem:random} using the Structural Lemma}
\label{sec:proof-random}
Let $G$ be a random graph and let $S$ be a random subset of vertices, both defined in the statement of \cref{lem:random}.
According to \cref{lem:B}, $\expect{G,S}{|\hat E\setminus E|}\leq \expect{G}{|B|}+o(1)$.
It suffices to prove that $\expect{G}{|B|}=o(1)$.

First, we consider the case when $\Delta=O(1)$ is such that $\Delta\geq 3$.
Our analysis is based on the configuration model.
Let $G'$ be a multigraph corresponding to a uniformly random configuration.
Let $\expect{G'}{|B|}$ denote the expected size of the set $B$ defined on $G'$.
Since each $\Delta$-regular graph corresponds to the same number of configurations and because the probability spaces of configurations and of $\Delta$-regular graphs, respectively, are uniform, we have
$\expect{G}{|B|}\leq \expect{G'}{|B|}/p$, where $p$ is the probability that $G'$ is both simple and connected.
According to~\cite{wormald1999models}, when $\Delta\geq 3$, $p\sim e^{(1-\Delta^2)/4}$, which is constant since $\Delta=O(1)$.
Thus $\expect{G}{|B|}=O(\expect{G'}{|B|})$.

In order to bound $\expect{G'}{|B|}$, consider any vertex pair $\{v,w\}$ in $G'$ such that $\delta(v,w)\geq 2$.
From the Structural Lemma (\cref{lem:Fvw}), the event $|D(v,w)|\leq 3n/\log n$ occurs with probability $o(n^{-2})$.
Equivalently, the event $|D(v,w)|\leq 3n\cdot (\log n)/s$ occurs with probability $o(n^{-2})$, since $s=\log^2 n$.
Thus the event $\{v,w\}\in B$ occurs with probability $o(n^{-2})$ according to the definition of $B$ in \cref{lem:B}.
There are $n(n-1)/2$ vertex pairs $\{v,w\}$ in $G'$.
By linearity of expectation, $\expect{G'}{|B|}$ is at most $o(n^{-2})\cdot n(n-1)/2=o(1)$.
Hence $\expect{G}{|B|}=O(\expect{G'}{|B|})=o(1)$.

In the special case when $\Delta=2$, a  2-regular graph $G$ is a ring.
Consider any vertex pair $\{v,w\}$ in $G$ such that $\delta(v,w)\geq 2$.
It is easy to see that at least $n-4$ vertices $u$ in the ring $G$ are such that $|\delta(u,v)-\delta(u,w)|> 1$, so $|D(v,w)|\geq n-4$ by \cref{def:distinguish}.
When $n$ is large enough, $n-4>3n/\log n$, so $|D(v,w)|>3n/\log n$.
Equivalently, we have $|D(v,w)|>3n\cdot (\log n)/s$, since $s=\log^2 n$.
Thus $\{v,w\}\notin B$ according to the definition of $B$ in \cref{lem:B}.
Therefore, $B=\emptyset$ and $\expect{G}{|B|}=0$.

We conclude that  $\expect{G}{|B|}=o(1)$ for any $\Delta=O(1)$. Thus $\expect{G,S}{|\hat E\setminus E|}\leq \expect{G}{|B|}+o(1)=o(1)$.

We complete the proof of \cref{lem:random}.

\begin{remark}
\label{remark:extra-work}
With more care in the construction of the sequences in \cref{lem:dispensable}, we can improve the bound in Property~\ref{prop:ell_M} of \cref{lem:dispensable} by 
\[L_M\geq \lceil \log_{\Delta-1} (3n/(\log \log n)^{2+\epsilon}) \rceil+2,\] for any $\epsilon>0$.
As a result, the range of $k$ in \cref{lem:number-interesting} can be extended to
$k\leq \lceil \log_{\Delta-1} (3n/(\log \log n)^{2+\epsilon}) \rceil+2$, and consequently, the event in \cref{lem:Fvw} can be replaced by $|D(v,w)|>3n/(\log \log n)^{2+\epsilon}$.
Therefore, \cref{lem:random} holds for $s=\log n \cdot(\log \log n)^{2+\epsilon}$.
This implies that the parameter $s$ in \cref{thm:random} can be reduced to $\log n \cdot(\log \log n)^{2+\epsilon}$.
\end{remark}

\section{Other Reconstruction Models (Proofs of \cref{cor:betweenness,cor:all-distances,cor:parallel})}
\label{sec:variant}
In this section, we study the reconstruction of random $\Delta$-regular graphs in the all-distances query model, in the betweenness query model, as well as in the parallel setting.

By extending the analysis from \cref{sec:random}, we observe that the set $\hat E$ constructed in \algoname (\cref{algo:main}) equals the edge set $E$ with high probability, in \cref{lem:single-iteration}.

\begin{lemma}
\label{lem:single-iteration}
Let $G$ be a uniformly random $\Delta$-regular graph with $\Delta=O(1)$.
Let $s=\log^2 n$.
Let $S\subseteq V$ be a set of $s$ vertices selected uniformly and independently at random from $V$.
With probability $1-o(1)$, $|\hat E|=\Delta n/2$.
In addition, the event $|\hat E|=\Delta n/2$ implies $\hat E=E$.
\end{lemma}
\begin{proof}
From \cref{lem:random}, $\expect{G,S}{|\hat E\setminus E|}=o(1)$.
By Markov's inequality, the event that $|\hat E\setminus E|\geq 1$ occurs with probability $o(1)$.
Thus with probability $1-o(1)$, we have $\hat E\subseteq E$.
On the other hand, $E\subseteq \hat E$ by \cref{fact:hat-E}.
Therefore, the event that $\hat E= E$ occurs with probability $1-o(1)$, and this event occurs if and only if $|\hat E|=|E|$.
The statement follows since $|E|=\Delta n/2$ in a $\Delta$-regular graph.
\end{proof}

\subsection{A Modified Algorithm}
\cref{lem:single-iteration} enables us to design another reconstruction algorithm in the distance query model, called \variantname, which is a modified version of \algoname, see \cref{algo:variant}.
\variantname repeatedly computes a set $\hat E$ as in \algoname, until the size of $\hat E$ equals $\Delta n/2$.
The parameter $s$ is fixed to $\log^2 n$.

\begin{algorithm}
\caption{\variantname$(V)$}
\label{algo:variant}
\begin{algorithmic}[1]
    \Repeat
    \State $S\gets$ sample of $s=\log^2 n$ vertices selected uniformly and independently at random from $V$
    \For{$u\in S$ and $v\in V$}
         \query$(u,v)$
    \EndFor
    \State $\hat E\gets$ set of vertex pairs $\{a,b\}\subseteq V$ such that, for all $u\in S$, $|\delta(u,a)-\delta(u,b)|\leq 1$
    \Until $|\hat E|=\Delta n/2$
    \State \Return $\hat E$
\end{algorithmic}
\end{algorithm}

\begin{lemma}
\label{lem:variant}
Let $G$ be a uniformly random $\Delta$-regular graph with $\Delta=O(1)$.
In the distance query model, \variantname (\cref{algo:variant}) is a reconstruction algorithm, i.e., its output equals the edge set $E$.
The expected number of iterations of the \textbf{repeat} loop in \variantname is $1+o(1)$.
\end{lemma}

\begin{proof}
Upon termination of the \textbf{repeat} loop in \variantname, we have $|\hat E|=\Delta n/2$, which implies $\hat E=E$ by \cref{lem:single-iteration}.
Thus the output of \variantname equals the edge set $E$.

In each iteration of the \textbf{repeat} loop, the event that $|\hat E|=\Delta n/2$ occurs with probability $1-o(1)$ by \cref{lem:single-iteration}.
Thus the expected number of iterations of the \textbf{repeat} loop is $1+o(1)$.
\end{proof}

\subsection{All-Distances Query Model (Proof of \cref{cor:all-distances})}
By \cref{lem:variant}, \variantname is a reconstruction algorithm in the distance query model.
We extend \variantname to the all-distances query model.

Observe that in \variantname, the distance queries are performed between each sampled vertex $u\in S$ and all vertices in the graph.
This is equivalent to a single query at each sampled vertex $u\in S$ in the all-distances query model. Hence each iteration of the \textbf{repeat} loop in \variantname corresponds to $|S|=\log^2 n$ all-distances queries.
Again by \cref{lem:variant}, the expected number of iterations of the \textbf{repeat} loop in \variantname is $1+o(1)$.
Therefore, in the all-distances query model, an algorithm equivalent to \variantname reconstructs the graph using $(1+o(1))\cdot \log^2 n=O(\log^2 n)$ all-distances queries in expectation.

\subsection{Betweenness Query Model (Proof of \cref{cor:betweenness})}
In the betweenness query model, Abrahamsen~et~al.~\cite{abrahamsen2016graph} showed that $\tilde O(\Delta^2\cdot n)$ betweenness queries suffice to compute the distances from a given vertex to all vertices in the graph (it is implicit in Lemma~16 from~\cite{abrahamsen2016graph}), so an all-distances query can be simulated by $\tilde O(\Delta^2\cdot n)$ betweenness queries.
As a consequence of \cref{cor:all-distances}, we achieve a reconstruction algorithm using $\tilde O(\Delta^2\cdot n\cdot \log^2 n)=\tilde O(n)$ betweenness queries in expectation, since $\Delta=O(1)$.

\subsection{Parallel Setting (Proof of \cref{cor:parallel})}
By \cref{lem:variant}, \variantname is a reconstruction algorithm in the distance query model.
We analyze \variantname in the parallel setting.

Each iteration of the \textbf{repeat} loop consists of $n\cdot \log^2 n$ distance queries, and the distance queries within the same iteration of the \textbf{repeat} loop can be performed in parallel.
Again by \cref{lem:variant}, the expected number of iterations of the \textbf{repeat} loop in \variantname is $1+o(1)$.
Thus the expected number of rounds in \variantname is $1+o(1)$, and the expected number of distance queries in \variantname is $(1+o(1))\cdot n \cdot\log^2 n=\tilde O(n)$.

\section{Metric Dimension (Proof of \cref{cor:metric_dimention})}
\label{sec:metric_dimension}
In this section, we study the metric dimension of random $\Delta$-regular graphs.
To begin with, we show an elementary structural property of random $\Delta$-regular graphs, in \cref{lem:metric-dimension}, based on a classical result on those graphs.

\begin{lemma}
\label{lem:metric-dimension}
Let $G=(V,E)$ be a uniformly random  $\Delta$-regular graph with $\Delta=O(1)$.
With probability $1-o(1)$, for any edge $(a,b)$ of the graph $G$, there exists a vertex $c\in V\setminus \{a,b\}$ that is adjacent to $b$ but is not adjacent to $a$.
\end{lemma}

\begin{proof}
First, consider the case when $\Delta=2$.
A $2$-regular graph is a ring.
Let $(a,b)$ be any edge of the graph.
The vertex $b$ has two neighbors, the vertex $a$ and another vertex, let it be $c$.
We have $c\in V\setminus \{a,b\}$ and $c$ is not adjacent to $a$ (as soon as $n>3$).
The statement of the lemma follows.

Next, consider the case when $\Delta=O(1)$ is such that $\Delta\geq 3$.
Let $\mathcal{E}$ denote the event that, for any edge $(a,b)$ of $G$, there do not exist two vertices $c_1$ and $c_2$ in $G$, such that all of the 4 edges $(a,c_1),(a,c_2),(b,c_1),(b,c_2)$ belong to $G$.
We show that $\mathcal{E}$ occurs with probability $1-o(1)$.
Indeed, if for some edge $(a,b)$ of $G$, there exist two vertices $c_1$ and $c_2$ such that $(a,c_1),(a,c_2),(b,c_1),(b,c_2)$ are edges of $G$, then the induced subgraph on $\{a,b,c_1,c_2\}$ consists of at least 5 edges.
A classical result on random $\Delta$-regular graphs shows that, for any constant integer $k$, the probability that there exists an induced subgraph of $k$ vertices with at least $k+1$ edges is $o(1)$, see, e.g., Lemma~11.12 in \cite{frieze_karonski_2015}.
Therefore, $\mathcal{E}$ occurs with probability $1-o(1)$.

We condition on the occurrence of $\mathcal{E}$.
For any edge $(a,b)$ of $G$, let $N(a)$ be the set of $\Delta-1$ neighbors of $a$ that are different from $b$, and let $N(b)$ be the set of $\Delta-1$ neighbors of $b$ that are different from $a$.
Since $\Delta\geq 3$, we have $|N(a)|=|N(b)|\geq 2$.
The event $\mathcal{E}$ implies that $N(a)\neq N(b)$, so there exists a vertex $c\in N(b)\setminus N(a)$.
By definition, $c$ is adjacent to $b$ but is not adjacent to $a$, and $c\in V\setminus \{a,b\}$.
Since $\mathcal{E}$ occurs with probability $1-o(1)$, we conclude that, with probability $1-o(1)$, for any edge $(a,b)$ of the graph $G$, there exists a vertex $c\in V\setminus\{a,b\}$ that is adjacent to $b$ but is not adjacent to $a$.
\end{proof}

\begin{definition}[e.g., \cite{bailey2011base,chartrand2000resolvability}]
\label{def:resolving}
A subset of vertices $S\subseteq V$ is a \emph{resolving set} for a graph $G=(V,E)$ if, for any pair of vertices $\{a,b\}\subseteq V$, there is a vertex $u\in S$ such that $\delta(u,a)\neq \delta(u,b)$.
The \emph{metric dimension} of $G$ is the smallest size of a resolving set for $G$.
\end{definition}

Based on the analysis of \algoname from \cref{lem:single-iteration} and  the structural property from \cref{lem:metric-dimension}, we show that, with high probability, a random subset of $\log^2 n$ vertices is a resolving set for a random $\Delta$-regular graph, in \cref{lem:landmarks}.

\begin{lemma}
\label{lem:landmarks}
Let $G=(V,E)$ be a uniformly random  $\Delta$-regular graph with $\Delta=O(1)$.
Let $S\subseteq V$ be a sample of $s=\log^2 n$ vertices selected uniformly and independently at random from $V$.
With probability $1-o(1)$, the set $S$ is a resolving set for the graph $G$.
\end{lemma}

\begin{proof}
Let $\mathcal{E}_1$ denote the event that, for any edge $(a,b)$ of the graph $G$, there exists a vertex $c\in V\setminus\{a,b\}$ that is adjacent to $b$ but is not adjacent to $a$.
By \cref{lem:metric-dimension}, the event $\mathcal{E}_1$ occurs with probability $1-o(1)$.
Let $\mathcal{E}_2$ denote the event $\hat E=E$.
By \cref{lem:single-iteration}, the event $\mathcal{E}_2$ occurs with probability $1-o(1)$.
Thus with probability $1-o(1)$, both events $\mathcal{E}_1$ and $\mathcal{E}_2$ occur simultaneously.
We condition on the occurrences of both events $\mathcal{E}_1$ and $\mathcal{E}_2$ in the subsequent analysis.

First, consider any vertex pair $\{a,b\}\subseteq V$ such that $\delta(a,b)\geq 2$.
The event $\mathcal{E}_2$ implies that $\{a,b\}\notin \hat E$.
By definition, there exists some vertex $u\in S$ such that $|\delta(u,a)-\delta(u,b)|\geq 2$, which implies that $\delta(u,a)\neq \delta(u,b)$.

Next, consider any vertex pair $\{a,b\}\subseteq V$ such that $\delta(a,b)=1$.
The event $\mathcal{E}_1$ implies that there exists a vertex $c\in V\setminus\{a,b\}$ that is adjacent to $b$ but is not adjacent to $a$. Since $\delta(a,c)\geq 2$, the event $\mathcal{E}_2$ implies that $\{a,c\}\notin \hat E$.
By definition, there exists some vertex $u\in S$ such that $|\delta(u,a)-\delta(u,c)|\geq 2$.
Using an elementary inequality of $|x-y|+|y-z|\geq |x-z|$ for any three real numbers $x$, $y$, and $z$, we have
\begin{align*}
|\delta(u,a)-\delta(u,b)|&\geq |\delta(u,a)-\delta(u,c)|-|\delta(u,b)-\delta(u,c)|\\
&\geq |\delta(u,a)-\delta(u,c)|-\delta(b,c)&\text{(by triangle inequality)}\\
&\geq 2-\delta(b,c) & \text{(by definition of $u$)} \\
&\geq 1 & \text{(since $(b,c)$ is an edge)}&.
\end{align*}
Thus $\delta(u,a)\neq \delta(u,b)$.

Therefore, conditioned on the occurrences of both events $\mathcal{E}_1$ and $\mathcal{E}_2$, for any vertex pair $\{a,b\}\subseteq V$, there exists a vertex $u\in S$ such that $\delta(u,a)\neq \delta(u,b)$.

We conclude that, with probability $1-o(1)$, the set $S$ is a resolving set for $G$.
\end{proof}

From \cref{lem:landmarks}, with probability $1-o(1)$, the metric dimension of a random $\Delta$-regular graph is at most $\log^2 n$.
This completes the proof of \cref{cor:metric_dimention}. 
\section{Reconstruction of Bounded-Degree Graphs (Proof of \cref{thm:general})}
\label{sec:general}
In this section, we analyze \algoname (\cref{algo:main}) on general graphs of bounded degree in the distance query model.
Recall that a set $B$ of vertex pairs $\{a,b\}\subseteq V$ is defined in \cref{lem:B}.
For every vertex $a\in V$, we define the set of vertices $B(a)\subseteq V$ as 
\[B(a)=\big\{ b\in V \mid \{a,b\}\in B\big\}.\]
Intuitively, $B(a)$ consists of the vertices $b\in V$ that has few distinguishers with $a$.
We bound the size of the set $B(a)$ for any vertex $a$, in \cref{lem:Ba}.

\begin{lemma}
\label{lem:Ba}
Let $G$ be a general  graph of bounded degree $\Delta$.
For any vertex $a\in V$, $|B(a)|\leq 9\Delta^3\cdot n^2\cdot(\log^2 n)/s^2$.
\end{lemma}

We defer the proof of \cref{lem:Ba} for the moment and first show how it implies \cref{thm:general}.

\begin{proof}[Proof of \cref{thm:general} using \cref{lem:Ba}]
By \cref{lem:query}, \algoname is a reconstruction algorithm using $n\cdot s+|\hat E|$ distance queries, and in addition, \algoname can be parallelized using 2 rounds.
It remains to further analyze the query complexity.

From \cref{fact:hat-E}, $|\hat E|=|E|+|\hat E\setminus E|$.
Since the graph has bounded degree $\Delta$, $|E|\leq \Delta n$.
From \cref{lem:B}, $\expect{S}{|\hat E\setminus E|}\leq |B|+o(1)$.
Therefore, the expected number of distance queries in \algoname is at most $n\cdot s+\Delta n+|B|+o(1)$.
It suffices to analyze $|B|$.

Observe that $|B|\leq \sum_{a\in V} |B(a)|$ by definition of $\{B(a)\}_{a\in V}$.
From \cref{lem:Ba}, $|B(a)|\leq 9\Delta^3\cdot n^2\cdot(\log^2 n)/s^2$, for any vertex $a\in V$.
Hence $|B|\leq (9\Delta^3\cdot n^2\cdot(\log^2 n)/s^2)\cdot n$.
Thus the expected number of distance queries in \algoname is at most $n\cdot s+ \Delta n+(9\Delta^3\cdot n^2\cdot(\log^2 n)/s^2)\cdot n+o(1)$, which is $\tilde O(n^{5/3})$ since $s=n^{2/3}$ and $\Delta=O(\polylog n)$.
\end{proof}

The rest of the section is dedicated to prove Lemma~\ref{lem:Ba}.

Let $a$ be any vertex in $V$.
Let $T$ be an (arbitrary) shortest-path tree rooted at $a$ and spanning all vertices in $V$.
For any vertex $b\in V$, let \emph{the shortest $a$-to-$b$ path} denote the path between $a$ and $b$ in the tree $T$.
To simplify the presentation, we assume that, for any $b\in B(a)$, $\delta(a,b)$ is even, so that the \emph{midpoint vertex} of the shortest $a$-to-$b$ path is uniquely defined.
We extend our analysis to the general setting in the end of the section.

For any vertex $m\in V$, define the set $B(a,m)\subseteq B(a)$ as
\[B(a,m)=\big\{b\in B(a) \mid \text{the midpoint vertex of the shortest $a$-to-$b$ path is $m$}\big\}.\]
Define the set $M(a)\subseteq V$ as
\[M(a)=\big\{m\in V\mid B(a,m)\neq \emptyset\big\}.\]
In other words, $M(a)$ consists of the vertices $m\in V$ that is the midpoint vertex of the shortest $a$-to-$b$ path for some $b\in B(a)$.
From the construction, we have
\begin{equation}
\label{eqn:Ba}
    B(a)=\bigcup_{m\in M(a)} B(a,m).
\end{equation}

In order to bound the size of $B(a)$, first we bound the size of $B(a,m)$ for any midpoint $m\in M(a)$,  in \cref{lem:simple-1}, and then we bound the number of distinct midpoints,  in \cref{lem:simple-2}.

\begin{lemma}
\label{lem:simple-1}
For any $m\in M(a)$, $|B(a,m)|\leq 3\Delta\cdot n\cdot(\log n)/s$.
\end{lemma}

\begin{proof}
For any $b\in B(a,m)$, the vertex $m$ is the midpoint vertex of the shortest $a$-to-$b$ path by definition.
From the assumption, $\delta(a,b)$ is even for any $b\in B(a,m)$, so there exists for some positive integer $\ell$, such that $\delta(m,a)=\ell$ and $\delta(m,b)=\ell$ for any $b\in B(a,m)$.

For every neighbor $m'$ of $m$ such that $\delta(a,m')=\delta(a,m)+1$, define a set $Y(m')\subseteq B(a,m)$ that consists of the vertices $b\in B(a,m)$ such that $m'$ is on the shortest $a$-to-$b$ path.
Let $\hat m$ be a neighbor of $m$ such that $\delta(a,\hat m)=\delta(a,m)+1$ and that $|Y(\hat m)|$ is maximized, see \cref{fig:simple-1}.
Since the graph has bounded degree $\Delta$, we have
$|B(a,m)|\leq \Delta\cdot|Y(\hat m)|$.
It suffices to bound $|Y(\hat m)|$.

The main observation is that any vertex of $Y(\hat m)$ distinguishes $a$ and any other vertex of $Y(\hat m)$.
To see this, let $b_0$ be any vertex in $Y(\hat m)$.
By definition, $\delta(a,\hat m)=\delta(a,m)+1=\ell+1$.
Since $\hat m$ is on the shortest $a$-to-$b_0$ path, we have $\delta(\hat m,b_0)=\delta(a,b_0)-\delta(a,\hat m)=\ell-1$, thus $\delta(\hat m,b_0)=\delta(\hat m,a)-2$.
For any vertex $b_1\in Y(\hat m)$, from the triangle inequalities on $\delta$, we have
\[\delta(b_1,b_0) \leq \delta(b_1,\hat m)+\delta(\hat m,b_0) =\delta(b_1,\hat m)+\delta(\hat m,a)-2 = \delta(b_1,a)-2.\]
According to \cref{def:distinguish}, the vertex $b_1$ distinguishes $a$ and $b_0$, and equivalently, $b_1\in D(a,b_0)$.
Thus we have $Y(\hat m)\subseteq D(a,b_0)$, hence $|Y(\hat m)|\leq |D(a,b_0)|\leq 3n\cdot (\log n)/s$ using the fact that $b_0\in Y(\hat m)\subseteq B(a)$ and the definition of $B$ in \cref{lem:B}.

We conclude that
$|B(a,m)|\leq \Delta\cdot |Y(\hat m)|\leq 3\Delta\cdot n\cdot(\log n)/s.$
\end{proof}

\begin{figure}[t]
\centering
\includegraphics[scale=0.675]{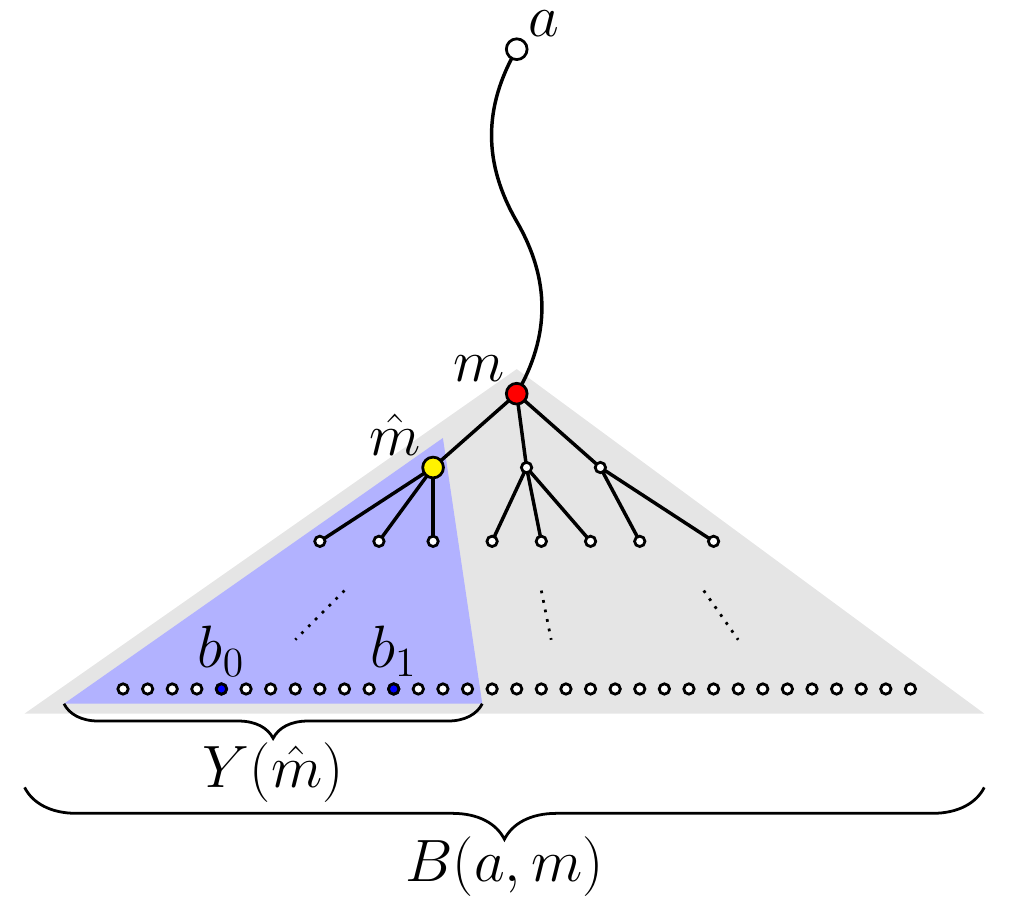}
\caption{
The vertex $m$ is the midpoint of the shortest path between $a$ and any vertex in $B(a,m)$.
The vertex $\hat m$ is a well-chosen neighbor of $m$.
Consider any vertex $b_0\in Y(\hat m)$.
We can show that any vertex $b_1\in Y(\hat m)$ distinguishes $a$ and $b_0$.
}
\label{fig:simple-1}
\end{figure}

\begin{lemma}
\label{lem:simple-2}
$|M(a)|\leq 3\Delta\cdot n\cdot(\log n)/s$.
\end{lemma}

\begin{proof}
For each vertex $m\in M(a)$, denote $x_m$ as the second-to-last vertex on the shortest $a$-to-$m$ path.
Denote $X(a)\subseteq V$ as the set of vertices $x_m$ for all $m\in M(a)$.
See \cref{fig:simple-2}.
Since $G$ has bounded degree $\Delta$, we have
$|M(a)|\leq \Delta\cdot |X(a)|$.
It suffices to bound $|X(a)|$.

Let $b^*$ be a vertex in $B(a)$ such that $\delta(a,b^*)$ is maximized.
From the assumption, $\delta(a,b^*)$ is even, so we denote $\delta(a,b^*)=2\ell$ for some positive integer $\ell$.

The main observation is that any vertex of $X(a)$ distinguishes $a$ and $b^*$.
To see this, let $x$ be any vertex in $X(a)$.
Let $m$ be any vertex in $M(a)$ such that $x$ is the second-to-last vertex on the shortest $a$-to-$m$ path.\footnote{Such a vertex $m$ exists according to the construction of $X(a)$.}
We have $\delta(a,m)\leq \ell$ and $\delta(a,x)=\delta(a,m)-1\leq \ell-1.$
By the triangle inequality on the distances, $\delta(b^*,x)\geq \delta(a,b^*)-\delta(a,x)\geq 2\ell-(\ell-1)=\ell+1$.
Thus $\delta(b^*,x)-\delta(a,x)\geq 2$.
According to \cref{def:distinguish}, the vertex $x$ distinguishes $a$ and $b^*$, and equivalently, $x\in D(a,b^*)$.
Thus $X(a)\subseteq D(a,b^*)$, hence $|X(a)|\leq |D(a,b^*)|\leq 3n\cdot (\log n)/s$ using the fact that $b^*\in B(a)$ and the definition of $B$ in \cref{lem:B}.

We conclude that $|M(a)|\leq \Delta\cdot |X(a)|\leq 3\Delta\cdot n \cdot (\log n)/s.$
\end{proof}

\begin{figure}[t]
\centering
\includegraphics[scale=0.675]{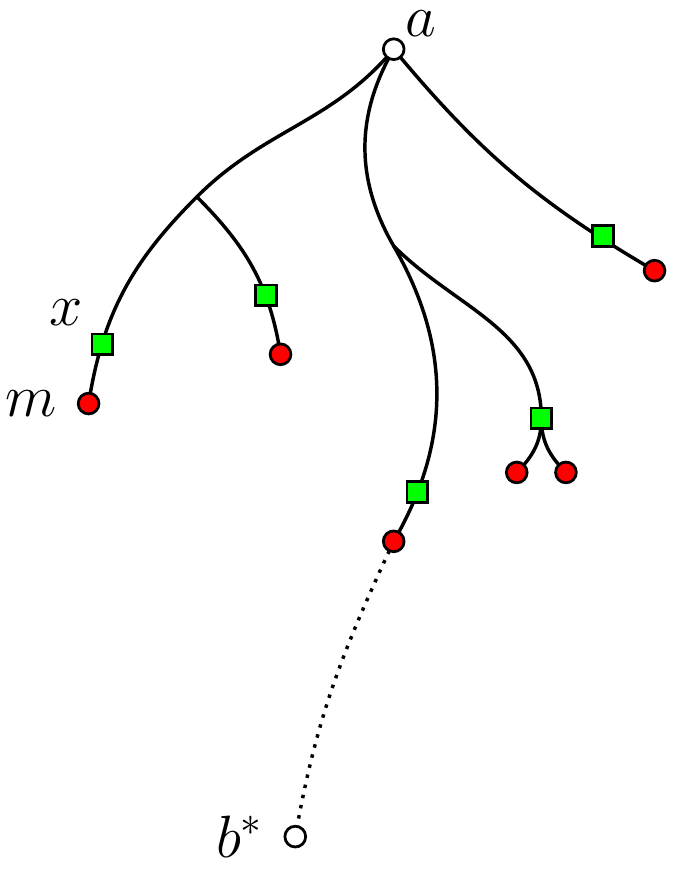}
\caption{
Solid circular nodes represent the vertices $m\in M(a)$.
Solid curves represent the shortest $a$-to-$m$ paths.
Solid square nodes represent the vertices in $X(a)$.
Denote $b^*$ as a vertex in $B(a)$ that is farthest from $a$.
We can show that any vertex $x\in X(a)$ distinguishes $a$ and $b^*$.
}
\label{fig:simple-2}
\end{figure}

From \cref{eqn:Ba}, $|B(a)|\leq \sum_{m\in M(a)}|B(a,m)|$.
From \cref{lem:simple-1}, $|B(a,m)|\leq 3\Delta\cdot n\cdot(\log n)/s$ for every $m\in M(a)$.
From \cref{lem:simple-2}, $|M(a)|\leq 3\Delta\cdot n\cdot(\log n)/s$.
Therefore, $|B(a)|\leq 9\Delta^2\cdot n^2\cdot(\log^2 n)/s^2$.

Finally, consider the general setting in which $\delta(a,b)$ is not necessarily even for any $b\in B(a)$.
For a vertex $m$ on the shortest $a$-to-$b$ path, we say that $m$ is the \emph{midpoint vertex} of that path if $\delta(a,m)=\lfloor\delta(a,b)/2\rfloor$.
The definitions of $B(a,m)$ and $M(a)$ remain the same.
\cref{lem:simple-2} holds in the same way.
In \cref{lem:simple-1}, the upper bound of $|B(a,m)|$ is replaced by $3\Delta^2\cdot n\cdot(\log n)/s$. Indeed, to extend the proof of \cref{lem:simple-1}, instead of considering vertex $m'$ (resp., vertex $\hat m$) that is a neighbor of $m$, we consider $m'$ (resp., $\hat m$) that is at distance 2 from $m$.
We have $|B(a,m)|\leq \Delta^2\cdot|Y(\hat m)|$.
The bound $|Y(\hat m)|\leq 3n\cdot (\log n)/s$ remains the same, so we have $|B(a,m)|\leq 3\Delta^2\cdot n\cdot(\log n)/s.$
Hence $|B(a)|\leq 9\Delta^3\cdot n^2\cdot(\log^2 n)/s^2$.

We complete the proof of \cref{lem:Ba}.
Therefore, we obtain \cref{thm:general}.

\bibliography{references}

\end{document}